\RequirePackage{fix-cm}
\documentclass[twocolumn]{svjour3}          

\smartqed  
\usepackage{graphicx}
\usepackage{microtype}					
\usepackage{xcolor}						
\usepackage{xparse}						

\usepackage[intlimits,fleqn]{amsmath}
\usepackage{amssymb}
\usepackage[safe]{tipa}					
\usepackage[misc]{ifsym}				
\usepackage{nicefrac}
\usepackage{bm}							
\usepackage{subfig}						
\usepackage{placeins}					

\usepackage[]{hyperref}		

\hypersetup{
	colorlinks,				
	citecolor=black,			
	filecolor=black,
	linkcolor=black,
	urlcolor=black
}

\IfFileExists{mathtools.sty}{\usepackage{mathtools}}{}

\usepackage[linesnumbered,vlined,boxed,oldcommands]{algorithm2e}
\SetKwComment{comment}{\% }{}
\IncMargin{8mm}

\newcommand{\mbf}[1]{{\ensuremath{\bm{{#1}}}}}

\newcommand{\abs}[1]{{\ensuremath{\left| #1 \right|}}}

\newcommand{\norm}[1]{\ensuremath{\left\| #1 \right\|}}

\newcommand{\MATLAB}{\ensuremath{\textrm{MATLAB}^\circledR}}

\let\hat\widehat

\makeatletter
\let\ams@underbrace=\underbrace
\def\underbrace{\kernel@ifnextchar[{\underbrace@}{\underbrace@[l]}}
\def\underbrace@[#1]#2_#3{%
  \ifx#1c\relax
    \let\ubr@align\centering%
  \else
    \ifx#1l\relax
      \let\ubr@align\raggedright%
    \else
      \ifx#1r\relax
        \let\ubr@align\raggedleft%
      \else
        \ifx#1f\relax
          \let\ubr@align\relax%
        \else
          \message{`#1' isn't a valid alignment specification for the underbrace command}%
        \fi
      \fi
    \fi
  \fi
  \setbox0=\hbox{$\displaystyle#2$}%
  \ams@underbrace{#2}_{\parbox[t]{\the\wd0}{\ubr@align\ensuremath{#3}}}%
}
\let\ubr@align\centering
\makeatother

\NewDocumentCommand\col{g}{%
  \IfNoValueTF{#1}{\ensuremath{\mathrm{vec}}}{\ensuremath{\mathrm{vec}}\of{#1}}%
}

\NewDocumentCommand\of{og}{%
  \IfNoValueTF{#1}%
    { \IfNoValueTF{#2}{}{\!\({#2}\)} }%
    { \IfNoValueTF{#2}{\!\[{#1}\]}{\!\{{#2}\}} }%
}

\DeclareMathOperator{\Diff}{\ipaclap{D}{\raisebox{.204em}{\textpalhook}\kern.44em}\kern-.1em}
\NewDocumentCommand\diff{g}{%
  \IfNoValueTF{#1}
  {\text{\texthtd}}
  {\text{\texthtd}\of{#1}}%
}

\RenewDocumentCommand\ln{g}{%
  \IfNoValueTF{#1}{\mathrm{ln\ }}{\mathrm{ln}\of{#1}}%
}

\NewDocumentCommand\Real{og}{%
  \IfNoValueTF{#1}%
    { \IfNoValueTF{#2}{\mathcal{R}\!\!\mathpzc{e}}{\mathcal{R}\!\!\mathpzc{e}\!\{{#2}\}} }%
    { \IfNoValueTF{#2}{\mathcal{R}\!\!\mathpzc{e}\!\[{#1}\]}{\mathcal{R}\!\!\mathpzc{e}\!\({#2}\)} }%
}
\NewDocumentCommand\Imag{og}{%
  \IfNoValueTF{#1}%
    { \IfNoValueTF{#2}{\mathcal{I}\!\!\mathpzc{m}}{\mathcal{I}\!\!\mathpzc{m}\!\{{#2}\}} }%
    { \IfNoValueTF{#2}{\mathcal{I}\!\!\mathpzc{m}\!\[{#1}\]}{\mathcal{I}\!\!\mathpzc{m}\!\({#2}\)} }%
}

\RenewDocumentCommand\cos{g}{%
  \IfNoValueTF{#1}{\mathrm{cos}}{\mathrm{cos}\of{#1}}%
}
\RenewDocumentCommand\sin{g}{%
  \IfNoValueTF{#1}{\mathrm{sin}}{\mathrm{sin}\of{#1}}%
}
\RenewDocumentCommand\tan{g}{%
  \IfNoValueTF{#1}{\mathrm{tan}}{\mathrm{tan}\of{#1}}%
}
\RenewDocumentCommand\arccos{g}{%
  \IfNoValueTF{#1}{\mathrm{arccos}}{\mathrm{arccos}\of{#1}}%
}
\RenewDocumentCommand\arcsin{g}{%
  \IfNoValueTF{#1}{\mathrm{arcsin}}{\mathrm{arcsin}\of{#1}}%
}
\RenewDocumentCommand\arctan{g}{%
  \IfNoValueTF{#1}{\mathrm{arctan}}{\mathrm{arctan}\of{#1}}%
}
\RenewDocumentCommand\cot{g}{%
  \IfNoValueTF{#1}{\mathrm{cot}}{\mathrm{cot}\of{#1}}%
}

\newcommand{\st}[1][14pt]{\text{\ \hspace{#1} s.t. \hspace{#1}}}

\NewDocumentCommand\tr{g}{%
  \IfNoValueTF{#1}{\mathrm{tr}}{\mathrm{tr}\of{#1}}%
}

\NewDocumentCommand\diag{og}{%
  \IfNoValueTF{#1}%
    { \IfNoValueTF{#2}{\ensuremath{\mathrm{diag}}}{\ensuremath{\mathrm{diag}\of{#2}}} }%
    { \IfNoValueTF{#2}{\ensuremath{\mathrm{diag}\of[#1]}}{\ensuremath{\mathrm{diag}\of[]{#2}}} }%
}

\RenewDocumentCommand\exp{g}{%
  \IfNoValueTF{#1}{\ensuremath{\mathrm{exp}}}{\ensuremath{\mathrm{exp}}\of{#1}}%
}

\NewDocumentCommand\realPart{g}{%
  \IfNoValueTF{#1}{\mathrm{Re}}{\mathrm{Re}\of{#1}}%
}
\NewDocumentCommand\imagPart{g}{%
  \IfNoValueTF{#1}{\mathrm{Im}}{\mathrm{Im}\of{#1}}%
}

\newcommand{\Operator}[3][]{\ensuremath{{\mathrm{#2}}_{#1}\!\left[ #3 \right]}}

\newcommand{\E}[2][]{\Operator[#1]{E}{#2}}

\newcommand{\Var}[2][]{\Operator[#1]{Var}{#2}}

\NewDocumentCommand\C{g}{%
  \IfNoValueTF{#1}{\mathrm{Cov}}{\mathrm{Cov}\of{#1}}%
}

\NewDocumentCommand\DKL{gg}{%
  \IfNoValueTF{#1}{\mathrm{D_{KL}}}{\mathrm{D_{KL}}\of{\left. #1\ \middle\|\ #2 \right.}}
  }

\NewDocumentCommand\ent{ggg}{%
  \IfNoValueTF{#1}{\mathrm{H}}{        	
  \IfNoValueTF{#2}{\mathrm{H}\of{#1}}   
  {\mathrm{H}\of{\left. #1\ \middle|\ #2 \right.}} } 
}

\renewcommand{\(}{\ensuremath{\left(}}
\renewcommand{\)}{\ensuremath{\right)}}
\renewcommand{\[}{\ensuremath{\left[}}
\renewcommand{\]}{\ensuremath{\right]}}

\let\oldBracketLeft\{
\let\oldBracketRight\}
\renewcommand{\{}{\ensuremath{\left\oldBracketLeft}}
\renewcommand{\}}{\ensuremath{\right\oldBracketRight}}

\DeclareMathOperator*{\argmax}{arg\,max}
\DeclareMathOperator*{\argmin}{arg\,min}

\DeclareMathOperator*{\gtlt}{\ensuremath{\begin{array}{c} > \vspace{-6pt} \\ < \end{array}}}

\NewDocumentCommand\F{og}{%
  \IfNoValueTF{#1}%
    { \IfNoValueTF{#2}{\mathcal{F}}{\mathcal{F}\!\{{#2}\}} }%
    { \IfNoValueTF{#2}{\mathcal{F}\!\[{#1}\]}{\mathcal{F}\!\({#2}\)} }%
}
\NewDocumentCommand\FInv{og}{%
  \IfNoValueTF{#1}%
    { \IfNoValueTF{#2}{\mathcal{F}^{-1}}{\mathcal{F}^{-1}\!\{{#2}\}} }%
    { \IfNoValueTF{#2}{\mathcal{F}^{-1}\!\[{#1}\]}{\mathcal{F}^{-1}\!\({#2}\)} }%
}

\NewDocumentCommand\rect{g}{%
  \IfNoValueTF{#1}
  {\ensuremath{\mathrm{rect}}}
  {\ensuremath{\mathrm{rect}\of{#1}}}%
}

\NewDocumentCommand\sinc{g}{%
  \IfNoValueTF{#1}
  {\ensuremath{\mathrm{sinc}}}
  {\ensuremath{\mathrm{sinc}\of{#1}}}%
}

\NewDocumentCommand\supp{g}{%
  \IfNoValueTF{#1}
  {\ensuremath{\mathrm{supp}}}
  {\ensuremath{\mathrm{supp}\of{#1}}}%
}

\DeclareMathOperator*{\defeq}{\stackrel{\mathrm{def}}{=}}

\let\oldMathcal\mathcal
\renewcommand{\mathcal}[1]{\ensuremath{\oldMathcal{#1}}}

\makeatletter
\def\foreach#1#2#3{%
  \@test@foreach{#1}{#2}#3,\@end@token
}

\def\@swallow#1{}

\def\@test@foreach#1#2{%
  \@ifnextchar\@end@token%
    {\@swallow}%
    {\@foreach{#1}{#2}}%
}

\def\@foreach#1#2#3,#4\@end@token{%
  #1{#2}{#3}%
  \@test@foreach{#1}{#2}#4\@end@token%
}
\makeatother

\NewDocumentCommand{\bd}{}{\mbf{d}}

\NewDocumentCommand{\bg}{}{\mbf{g}}

\NewDocumentCommand{\br}{}{\mbf{r}}
\NewDocumentCommand{\bs}{}{\mbf{s}}

\NewDocumentCommand{\bu}{}{\mbf{u}}

\NewDocumentCommand{\bx}{}{\mbf{x}}

\NewDocumentCommand{\bD}{}{\mbf{D}}

\NewDocumentCommand{\bG}{}{\mbf{G}}
\NewDocumentCommand{\bH}{}{\mbf{H}}
\NewDocumentCommand{\bI}{}{\mbf{I}}

\NewDocumentCommand{\bdelta}{}{\mbf{\delta}}

\NewDocumentCommand{\bmu}{}{\mbf{\mu}}

\NewDocumentCommand{\btheta}{}{\mbf{\theta}}

\NewDocumentCommand{\bvarepsilon}{}{\mbf{\varepsilon}}

\NewDocumentCommand{\bSigma}{}{\mbf{\Sigma}}
\NewDocumentCommand{\bTheta}{}{\mbf{\Theta}}

\NewDocumentCommand{\bzero}{}{\mbf{0}}

\journalname{Journal of Mathematical Imaging and Vision}
\begin{document}

\title{Testing that a Local Optimum of the Likelihood is Globally Optimum using Reparameterized Embeddings
}
\subtitle{Applications to Wavefront Sensing}


\author{Joel W. LeBlanc$^{1,2}$ \and
        Brian J. Thelen$^{1,2}$ \and
		Alfred O. Hero$^{1}$}


\institute{
\Letter\ Joel W. LeBlanc$^{1,2}$ \at
\email{jwleblan@umich.edu}           
	\and
$^1$ University of Michigan \at
Ann Arbor, Michigan, 48109, USA
	\and
$^2$ Michigan Tech Research Institute \at
Ann Arbor, Michigan, 48105, USA
}

\date{Received: date / Accepted: date}

\authorrunning{Joel W. LeBlanc et al.}

\maketitle

\begin{abstract}
Many mathematical imaging problems are posed as non-convex optimization problems. When numerically tractable global optimization procedures are not available, one is often interested in testing ex post facto whether or not a locally convergent algorithm has found the globally optimal solution. When the problem is formulated in terms of maximizing the likelihood function under a statistical model for the measurements, one can construct a statistical test that a local maximum is in fact the global maximum.  A one-sided test is proposed for the case that the statistical model is a member of the generalized location family of probability distributions, a condition often satisfied in imaging and other inverse problems. We propose a general method for improving the accuracy of the test by reparameterizing the likelihood function to embed its domain into a higher dimensional parameter space. We show that the proposed global maximum testing method results in improved accuracy and reduced computation for a physically-motivated joint-inverse problem arising in camera-blur estimation. 


\end{abstract}

\def\bzero{{\mbf{0}}}
\section{Introduction} \label{sec:intro}

Mathematical imaging problems are often formulated as non-convex energy minimization problems that impose desirable properties on the global optima, e.g., corresponding to a denoised, deblurred, or segmented image. Much of the work of Mila Nikolova addressed the problem of local and global optima. As stated succinctly in one of her early papers: ``The resultant ... energy generally exhibits numerous local minima. Calculating its local minimum, placed in the vicinity of the maximum likelihood estimate, is inexpensive but inadequate'' \cite{nikolova1999markovian}.  Study of local and global optima was a recurring theme in her work, in which she addressed the nature of objective functions associated with non-convex probabilistic models, i.e., maximum likelihood (ML) and maximum a posteriori (MAP) \cite{nikolova1997estimees}, \cite{nikolova1998inversion}, \cite{nikolova1999markovian}, \cite{nikolova2000segmentation}, \cite{alberge2006blind}, \cite{nikolova2007model}, as well as non-linear least squares \cite{durand2006stabilityI}, \cite{durand2006stabilityII}.  Some of the optimization algorithms she introduced were only shown to converge to one of several possible local optima.  For such algorithms, an important question is whether an observed convergent limit is, in fact, the global maximum.  Searching for and identifying the global maximum is the problem that we address in this paper.          

We consider this problem in the general setting of maximum likelihood parameter estimation from multiple samples from a probability distribution that belongs to a parametric family. The conceptual simplicity and tractability of the Maximum Likelihood (ML) principle, along with its theoretical optimality properties, has made ML approaches prevalent in many fields. Yet, questions surrounding its practical application remain open. The pioneering statistician Sir Ronald Fisher \cite{fisher:1925} was an early advocate of the ML approach and is generally credited with its development, although similar concepts predate Fisher's work. Stigler \cite{stigler:2007} provides a historical account of the theory's maturation throughout the nineteenth and twentieth centuries. Asymptotic (large sample) characterization of local vs. global optima of the likelihood function was established by Le Cam (c.f. \cite{le-cam:2012a}, ch. 6) using central limit theory, but can be challenging to apply in practice. As the sample size increases, it has long been known that, under mild smoothness conditions, all statistically consistent stationary points of the likelihood function converge with probability one to the global maximum \cite{cramer:1946,wald:1949}. The natural question then becomes: when there exist multiple stationary points, and the number of samples is finite, how can one identify the globally optimal one?

There exist general-purpose algorithms to address this question, e.g., simulated annealing and genetic algorithms.  These algorithms, however, are rarely applied to high-dimensional problems because of high computational demands \cite{andrieu:2000,andrieu:2001,sharman:1989}. Stationary points of the likelihood function can be readily found using iterative root-finding methods such as Quasi-Newton gradient descent \cite{nocedal:1999}. Once a stationary point is found, it would be useful to have access to a simple test to determine if it is globally optimal without knowing the maximum value of the likelihood function. Several such tests have been proposed for this purpose \cite{biernacki:2005},\cite{blatt:2007}. In this paper, the focus is on testing local maxima of the likelihood function in the context of high dimensional inverse problems arising in signal processing and imaging. 

Specifically, this paper makes the following contributions. Starting with the global maximum validation function introduced by Biernacki \cite{biernacki:2005}, we demonstrate that its mean is always less than or equal to zero when the likelihood function belongs to a generalized location family of distributions: distributions parameterized by a shift in location. This property provides the impetus for constructing a one-sided variant of the test. This generalized location family is relevant to many linear and non-linear inverse problems. Furthermore, we introduce a new approach of testing for the global maximum by expanding the parameter space to a higher dimension through a reparameterized embedding and defining an augmented validation function for testing local maxima. The augmented validation function can better discriminate between local and global maxima due to the expanded parameter space. We provide a computational procedure for identifying useful candidate embeddings that significantly improves the accuracy of the test. Significant accuracy and computational advantages are demonstrated for the application of camera blur-function estimation \cite{leblanc:2018}. In particular, when implemented as a multiple restart stopping criterion, the proposed global maximum testing procedure is shown to significantly reduce computation as compared to simulated annealing methods for finding global maxima.  

%
%
The remainder of the paper is organized as follows. Section~\ref{sec:description} describes the general global maximum testing problem and introduces a simple illustrative example used to demonstrate the key concepts. Section~\ref{sec:globalOptim} introduces the one-sided global maximum test as a variant of the two-sided test of Biernacki \cite{biernacki:2005}, develops the reparameterized embedding method, and proposes a numerical spectral embedding procedure for identifying good embeddings. These concepts and procedures are illustrated in the context of a simple non-linear maximum likelihood estimation example. Section~\ref{sec:wavefront} illustrates the proposed methods for application to camera blur-function estimation. 

\section{Problem Description}\label{sec:description}
\subsection{Background}
The problem setting is as follows.  The observed data comes in the form of a matrix $\bd=[\bd_1, \ldots, \bd_n]\in \mathbb R^{m\times n}$ where  the columns are independent and identically distributed (i.i.d.) realizations of an $m$-dimensional random vector $\bD_1$ having a parametric Lebesgue density $f\of{\bd_1;\btheta}$.  The joint probability distribution $f(\bd,\btheta)=\prod_{k=1}^n f(\bd_k,\btheta)$ of $\bD=[\bD_1, \ldots, \bD_n]$ is a known function of $\bd$ and $\btheta$, and is in a parametric family $\{f(\bd;\btheta): \bd\in \mathcal D, \btheta \in \bTheta\}$. Here the vector of parameters $\btheta$ is unknown, taking values in a parameter space $\bTheta$, an open subset of $ \mathbb R^p$, and the matrix of measurements $\bd$ takes values in a sample space $\mathcal{D} \subset \mathbb R^{m\times n}$.    For any mean square integrable function $X$ of $\bD$ we define the statistical expectation $E_{\btheta_0}[X(\bD)]=\int X(\bd)f\of{\bd;\btheta_0}d\mu(\bd)$, where $d\mu(\bd)$ indicates integration with respect to the Lebesgue measure on $\mathbb R^m$. The subscript $\btheta_0$ of the expectation operator $E_{\btheta_0}$ is called the {\em true value} of $\btheta$ as it parameterizes the underlying density $f(\bd;\btheta_0)$ generating the observations. 

As the columns of the measurement matrix $\bd$ are i.i.d. realizations, the log-likelihood function is
\begin{align}
\ell\of{\bd ;\btheta} = \frac{1}{n}\sum_{k=1}^n \ln f\of{\bd_k;\btheta}
\end{align}
The associated maximum likelihood estimator (MLE) $\hat\btheta:\mathcal{D} \rightarrow \bTheta$ is defined as the global maximum
\begin{align}
\hat\btheta_{\mathrm{Global}} = \argmax_{\btheta\in \bTheta} \ell\of{\bd;\btheta}.\label{eqn:MLE}
\end{align}
An estimator $\hat\btheta$ is said to be consistent (statistically consistent in norm) when $\lim_{n\rightarrow\infty}E_{\btheta_0}[\|\hat\btheta-\btheta_0\|^2 ] \rightarrow 0$.  We will assume that $f$ is continuously differentiable in $\btheta$ for all $\bd$, and define the score function as $\bs\of{\bd,\btheta}=\nabla_{\btheta} \ell\of{\bd,\btheta}$. The Fisher information matrix \\$\bI\of{\btheta_0} = E_{\btheta_0}[\bs\of{\bD,\btheta_0}\bs\of{\bD,\btheta_0}^T]$ is assumed to exist and be invertible, and the notation ``$\overset{P}{\to}$''and ``$\overset{D}{\to}$'' will be used to describe convergence in probability and distribution respectively. 

For the problem addressed in this paper,  the global maximum of the log-likelihood is unknown, and only a local maximum $\hat\btheta$ is available, which is not necessarily equal to $\hat\btheta_{\mathrm{Global}}$. For example, the local maximum could be the limit of a convergent gradient descent algorithm.  Given $\hat\btheta$,  the local maximum testing problem is to decide between the two hypotheses
\begin{align}
H_0: \hat\btheta=\hat\btheta_{\mathrm{Global}} \quad vs. \quad H_1:\hat\btheta \neq \hat\btheta_{\mathrm{Global}}.\label{eqn:globalHypothesisTest}
\end{align}
A test between $H_0$ and $H_1$ is defined as a binary valued function $\phi: \mathcal D \rightarrow \{0,1\}$ that maps the data $\bd $ to $0$ or $1$, indicating the decision $H_0$ or $H_1$, respectively.  
The accuracy of a test is measured by its probability of false alarm PFA=$E_{\btheta_0}[\phi|H_0]$ and its probability of detection PD=$E_{\btheta_0}[\phi|H_1]$.  If for two tests $\phi_1$ and $\phi_2$ having identical PFA,  PD of $\phi_1$ is greater than PD of $\phi_2$, then $\phi_1$ is said to be more powerful than $\phi_2$. 

Many approaches to the general hypothesis testing problem (\ref{eqn:globalHypothesisTest}) have been studied over the years. Blatt and Hero \cite{blatt:2007} presented a historical context, which is summarized here. The likelihood ratio test \cite{wilks:1938}, Wald test \cite{wald:1943}, and Rao score test \cite{rao:1948} are asymptotically equivalent tests as the number $n$ of samples approaches infinity. The likelihood ratio and Wald tests require the distribution under $H_0$ to be known, which for (\ref{eqn:globalHypothesisTest}) requires knowledge of the true parameter. On the other hand, the Rao score test, later independently discovered and popularized under the name Lagrange multiplier test \cite{silvey:1959}, can be implemented when the true parameter is unknown. Rao's test measures the Euclidean norm of the score function weighted by the inverse Fisher information evaluated at a local maximum $\xi_R = \frac{1}{p}\bs\of{\bd,\hat\btheta}^T \bI^{-1}\of{\hat\btheta}\bs\of{\bd,\hat\btheta}$. Gan and Jiang \cite{gan:1999} propose a similar test for consistency of a stationary point of the log-likelihood based on White's information test \cite{white:1982}. White's original work was concerned with testing for model misspecification under the assumption that the global maximum of the likelihood function had been located, and Gan uses the same test statistic but in the converse situation.   

The Rao test may be used to test for consistency of a local maximum of the log-likelihood function.  Unfortunately, Monte Carlo experiments indicate that this test may not be very powerful even in the univariate setting \cite{gan:1999,biernacki:2005}.  In  \cite{biernacki:2005} an improved test was proposed for testing consistency of a stationary point following ideas presented by Cox \cite{cox:1961,cox:1962}.  This test, called the Biernacki test, uses a bootstrap estimate to directly compare the observed value of the locally maximized log-likelihood to its statistical expectation.  Both the Rao score and the Biernacki tests fall under the more general M-testing framework described by Blatt and Hero \cite{blatt:2007}, where additional types of tests of local maxima are proposed.  

\subsection{Motivating Example}\label{sec:motivatingExample}
\label{sec:simpleModelD}
To illustrate the difficulties in testing local maxima of the log-likelihood 
consider the following one dimensional statistical estimation problem. Let $x\in [0,T]$ be a time interval and $x_i=iT/N$, $i=0, \ldots, N-1$.  The measurements $\{d(x_i)\}_{i=1}^N$ are a set of time samples of a sinusoidal signal in additive Gaussian noise 
\begin{align}
d(x) &= \sin(\theta_0 x)+\varepsilon(x), \hspace{0.1in} x \in [0,T]
 \label{eqn:simpleModelD}.
\end{align}
Here $\theta_0 \geq 0$ is an unknown sinusoidal frequency parameter to be estimated. 
More generally, it will be more convenient to express the measurement model in vector form
$\bd=[d(x_0), \ldots, d(x_{N-1})]^T$
\begin{align}
\bd &= \bmu\of{\theta_0} + \bvarepsilon, \hspace{0.1in} \bvarepsilon \sim \mathcal{N}\of{\mbf{0},\sigma^2\bI_{N\times N}} \label{eqn:simpleModelA},
\end{align}
where $\bmu(\theta_0)=\sin(\theta_0 \bx)$ is the mean of $\bd$, a vector of time samples of the noiseless signal,  and $\bvarepsilon$ is an independent identically distributed (i.i.d.) zero mean $N$-dimenstional Gaussian noise with identity covariance matrix scaled by the variance parameter $\sigma^2$, which is assumed known. 
%
%
The maximum likelihood estimator $\hat\theta$ of the frequency parameter is then the globally optimal solution 
\begin{align}
\hat\theta_{\mathrm{Global}} &\defeq \argmin_{\theta\geq 0} \norm{\bmu\of{\theta}-\bd}^2. 
\label{eqn:simpleProblem}
\end{align}
For the sinusoidal signal in noise model (\ref{eqn:simpleModelD}), there will be local maxima of the log-likelihood function corresponding to the multiple stationary points of $\norm{\bmu\of{\theta}-\bd}^2$ over $\theta$. Figure~\ref{fig:exampleSignals}(a) shows two of these local maxima for the case that the noise variance $\sigma^2$ is zero. The solid curve in Figure~\ref{fig:exampleSignals}(a) corresponds to the true signal $\bmu(\theta)$, where $\theta=\theta_0$ is the global maximum, and the dashed curve corresponds to another signal for which $\theta$ is a local (non-global) maximum. Also shown is noisy data $\bd$ generated with each of these two signals, corresponding to blue circles and red crosses, respectively, where the Gaussian noise variance is $\sigma^2=1$. The difficulty in perceiving differences between these two noisy signals suggests that distinguishing a sub-optimal local maximum from the global maximum will be challenging.   This becomes even more difficult for higher dimensional estimation problems occurring in imaging (c.f.~\cite{leblanc:2018}~Section~3c).     
%
\begin{figure}[ht]
\centering
\captionsetup{position=top,labelfont=bf,singlelinecheck=off,justification=raggedright}
\subfloat[]{\hspace{-2mm}\includegraphics[width=0.92\columnwidth]{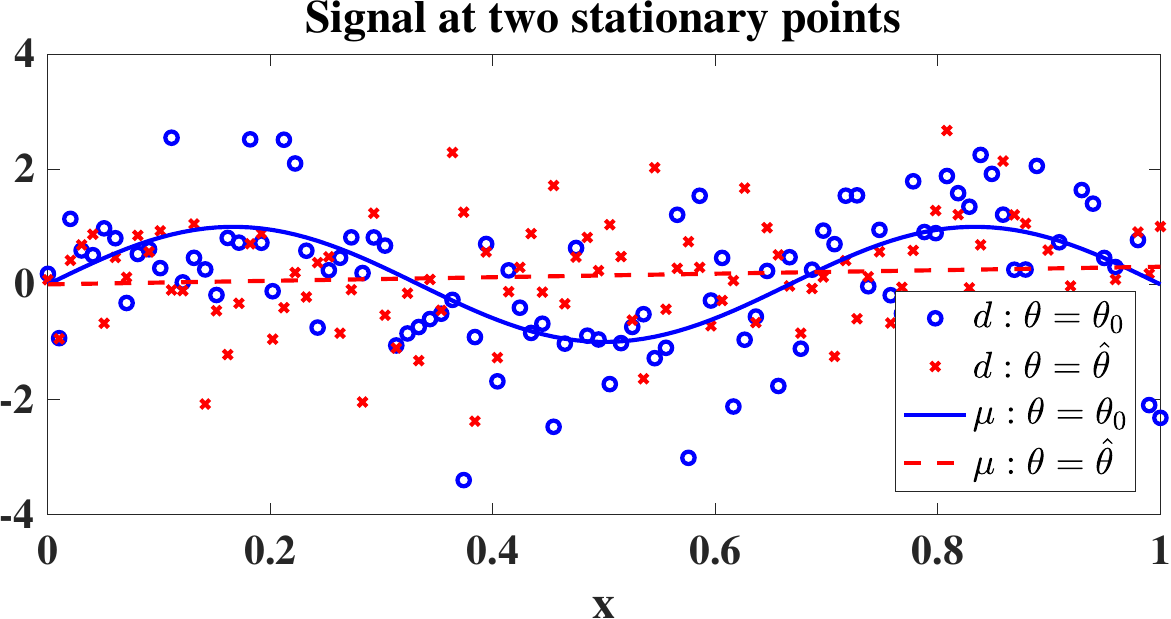}} \\ \vspace{-1em}
\subfloat[]{\hspace{-4mm}\includegraphics[width=0.95\columnwidth]{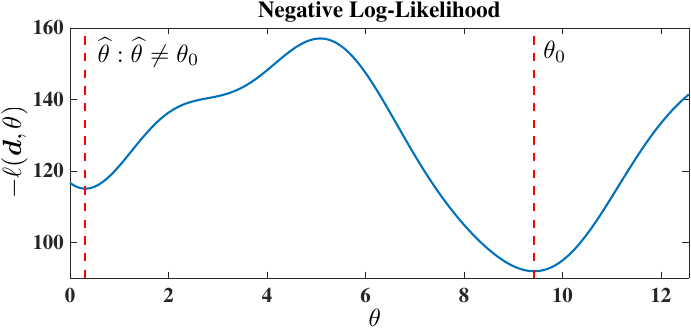}}
\caption{Realizations from a sinusoidal signal in Gaussian noise model for two values $\theta_0$ and $\hat\theta$ of the sinusoidal frequency parameter $\theta$ corresponding to a global maximum $\theta_0$ and a local minimum $\hat\theta$, respectively, of the log-likelihood $-\ell(\bd;\theta)$  (\ref{eqn:simpleProblem}) with $\sigma^2=1$. (a) Signal realizations (blue and red symbols) and the mean signal (blue and red curves) from the model when $\theta=\theta_0$ and $\theta=\hat\theta$, respectively.   (b) The negative log-likelihood plotted as a function of $\theta$ for the case $\sigma^2=0$. }
\label{fig:exampleSignals}
\end{figure}

\section{Tests for Local Optima}\label{sec:globalOptim}
In the hypothesis testing problem described by (\ref{eqn:globalHypothesisTest}), the null-hypothesis $H_0$ is that the discovered local maximum $\hat\btheta$ of the log-likelihood is a global maximum. It is important to note that a failure to reject the null hypothesis is not a positive statement about the global optimality of $\hat\btheta$. Instead, when a test accepts $H_0$, all that can be said is that it does not rule out the point as a local maximum with sufficient statistical certainty.

%

\def\bfm{{\mathbf m}}
\subsection{A Two-Sided Test}\label{sec:oneSidedTest} 
To test whether a local maximum of the likelihood function is, in fact, the global maximum one defines a suitable {\em global maximum validation function}  whose statistical distribution changes depending on whether the local maximum $\hat\btheta$ is global or not \cite{blatt:2007}.
Define the validation function 
\begin{align}
\varphi\of{\bd,\hat\btheta} \defeq \ell\of{\bd;\hat\btheta} - m(\hat\btheta,\hat\btheta),
\label{eq:validation}
\end{align}
where $m:\mathbb
R^p \times \mathbb R^p\rightarrow \mathbb R$ is the mean function
\begin{align}
    m(\btheta_0,\btheta_1)=\E[\btheta_0]{\ell\of{\bD;\btheta_1}},
    \label{eq:mdef}
\end{align}
This function is called the {\em ambiguity function} and is the statistical expectation under the distribution $f(\bd;\btheta_0)$ of the log-likelihood function $\ell(\bD,\btheta)$ evaluated at $\btheta=\btheta_1$. Assuming that the global maximum $\hat\btheta_{\mathrm{Global}}$ is near the true value  $\btheta_0$, under the null hypothesis $H_0$ we have $\hat\btheta=\hat\btheta_{\mathrm{Global}}$, and the distribution of $\varphi\of{\bD,\hat\btheta}$ will have approximately zero mean. On the other hand,  under the alternative hypothesis $H_1$ that $\hat\btheta$ is a non-global local maximum the mean of the distribution of $\varphi\of{\bD,\hat\btheta}$ will shift away from zero.   This is the key motivation for using the validation function (\ref{eq:validation}) to test for a global maximum.


For an i.i.d. data sample $\bD_1, \ldots, \bD_n$ the following asymptotic result was established in  \cite[Theorem 2]{biernacki:2005}. Under $H_0$: 
\begin{align}
\frac{1}{\sqrt{n}}\sum_{k=1}^n\varphi\of{\bD_k,\hat\btheta}  \overset{D}{\to} \mathcal{N}\of{0, \Var[\btheta_0]{\ell\of{\bD_1,\btheta_0}} },\label{eqn:H0Normal}
\end{align}
where $\Var{\ell\of{\bD_1,\btheta_0}}$ is the variance of the log-likelihood function for a single data sample ($n=1$).
Recalling the definition of the random data matrix $\bD=\{\bD_1, \ldots, \bD_n\}$,  this Gaussian limit motivates us to define the following test of the hypotheses $H_0: \hat{\btheta} = \hat{\btheta}_{Global}$ vs. $H_1: \hat{\btheta}\not=\hat{\btheta}_{Global}$  
\begin{align}
\frac{\(\ell\of{\bD;\hat\btheta}-m(\hat\btheta,\hat\btheta)\)^2} {v(\hat\btheta)}\; \; \overset{H_1}{\underset{H_0}\gtlt} \; \; \eta, \label{eqn:Biernacki}
\end{align}
where the function $v:\mathbb R^p \rightarrow \mathbb R$ is the variance $v(\theta)=\Var[\btheta]{\ell\of{\bD;\btheta}}$ under the distribution $f(\bd,\btheta)$ of the log-likelihood evaluated at $\btheta$, and  $\eta$ is a threshold selected to fix the false alarm probability equal to a suitably small number $\alpha\in[0,1]$. Under local asymptotically normal (LAN) conditions on the likelihood function  \cite{lehmann:1998} $\hat\btheta \overset{P}{\to}\btheta_0$ (a.s.) and the test statistic on the left hand side of (\ref{eqn:Biernacki}) has an approximately chi-square distribution under $H_0$. Hence $\eta$ can be selected as the $1-\alpha$ quantile of the chi-square distribution. In \cite{biernacki:2005} this test was implemented by approximating the mean $\E[\btheta_0]{\ell\of{\bD;\btheta_0}}$ and the variance $\Var{\ell\of{\bD;\btheta_0}}$ using a parametric bootstrap estimator.

The test (\ref{eqn:Biernacki}) is called a two-sided test because the condition for which the null hypothesis is accepted can be equivalently be expressed as
\begin{align*}
-\sqrt{\eta v(\hat\btheta)} \leq  \varphi\of{\bD,\hat\btheta}\leq \sqrt{\eta v(\hat\btheta)}.
\end{align*}
This is thus a test for which, as compared to the global maximum $\hat\btheta_{\mathrm{Global}}$, a sub-optimal local maximum $\hat\btheta$ will cause the test function to undergo a shift in mean, where the shift could either be in a positive or a negative direction.  

\subsection{A One-Sided Test}
If it were known {\em a priori} that a sub-optimal local maximum causes a negative shift in the mean of the global maximum validation function $\varphi\of{\bD,\hat\btheta}$, a one-sided test would be advantageous over a two-sided test. More specifically, a one-sided test would be expected to have higher power than the two-sided test (\ref{eqn:Biernacki}) when for all $\hat\btheta \neq \btheta_0$,
\begin{align}
m(\hat\btheta,\hat\btheta) \geq m(\btheta_0,\hat\btheta)
\label{eqn:oneSidedInequality},
\end{align}
where $m(\btheta_0, \btheta_1)$ is the ambiguity function defined in (\ref{eq:mdef}). 
When this condition is satisfied the two-sided test (\ref{eqn:Biernacki}) can be replaced by the one-sided test
\begin{align}
\frac{\ell\of{\bD;\hat\btheta} -m(\hat\btheta,\hat\btheta)} {\sqrt{v(\hat\btheta)}} \;\; \overset{H_0}{\underset{H_1}\gtlt} \;\; \eta_1.\label{eqn:rightTest}
\end{align}
%


The condition (\ref{eqn:oneSidedInequality}) is satisfied for many imaging and inverse problems. For example, consider the case where $\btheta$ is a clean image that one wishes to recover from samples of $\bD$, the output of an imaging sensor with known point spread function (forward operator) in additive correlated noise. When the point-spread function (PSF) and the covariance are known, we will show that this model always satisfies the inequality (\ref{eqn:oneSidedInequality}), and the one-sided test might be expected to lead to a better test for a global maximum.  Define $\btheta_0 \in \mathbb R^p$ as the vectorized true image to be recovered and $\bD\in \mathbb R^q$ as the random vectorized image acquired from the camera, which obeys the model: 
\begin{align}
    \bD = \bH \btheta_0 + \varepsilon \st \varepsilon \sim \mathcal{N}\of{\mbf{0},\bSigma} \label{eqn:inverseP},
\end{align}
where $\bH$ is a $q \times p$ matrix representing the forward operator and $\bSigma$ is the $q\times q$ camera noise covariance matrix. 

To show that (\ref{eqn:oneSidedInequality}) holds in this case, start with the log-likelihood function for the above model
\begin{align}
\ell\of{\bD;\btheta} =& -\frac{1}{2}\(\bH\(\btheta_0-\btheta\)+\varepsilon\)^T\bSigma^{-1} \(\bH\(\btheta_0-\btheta\)+\varepsilon\) \nonumber\\
	& \hphantom{-!} - \frac{1}{2}\ln{\det{\bSigma}} - \frac{q}{2}\ln{2\pi}. \label{eq:ncChi}
\end{align}
%
For any value of $\btheta$, (\ref{eq:ncChi}) is a quadratic form in $\varepsilon$ that is distributed non-central chi-squared with non-centrality parameter
\begin{align}
&\lambda = \(\btheta_0-\btheta\)^T \bH^T\bSigma^{-1}\bH \(\btheta_0-\btheta\).
\end{align}
The moment properties of the non-central chi-squared distribution \cite{kotz2004continuous} thus specify the statistical expectation of the log-likelihood function (\ref{eq:ncChi}) :
\begin{align}
\E[\btheta_0]{\ell\of{\bD;\btheta}} = -\frac{1}{2}\(q+\lambda\) - \frac{1}{2}\ln{\det{\bSigma}} - \frac{q}{2}\ln{2\pi}.\label{eqn:EEllNormal}
\end{align}
The difference $\E[\btheta]{\ell\of{\bD;\btheta}}-\E[\btheta_0]{\ell\of{\bD;\btheta}} = m(\btheta,\btheta)- m(\btheta_0,\btheta) = \lambda/2$, is non-negative, establishing that (\ref{eqn:oneSidedInequality}) holds as claimed.  
For this example, the unconstrained maximum likelihood estimator of $\btheta$ is a solution to a convex optimization problem, which is strictly convex when $\bH$ is full column-rank, and thus there will be no sub-optimal isolated local maxima of (\ref{eqn:MLE}).  
As our simple example in Figure \ref{fig:exampleSignals} illustrated, additional constraints can give rise to local maxima.
 
The condition (\ref{eqn:oneSidedInequality}) is satisfied for a more general class of camera models where the probability distribution of the data is in the generalized location family. 

\begin{definition}
Let $f(\bd;\btheta)$ be a distribution defined on $\bd\in \mathbb R^m$ parameterized by $\btheta\in \bTheta\subset \mathbb R^p$. The distribution belongs to the {\it generalized location family} of distributions if there exists a function $\bg:\mathbb R^p\rightarrow \mathbb R^m$ such that  $f\of{\bd;\btheta} = f\of{\bd-\bg\of{\btheta}}$ for all $\btheta \in \Theta$ and all $\bx \in \mathbb R^m$.
\end{definition}

Any camera model of the form $\bD_k=\bmu(\btheta_0)+\bvarepsilon_k$, $k=1, \ldots, n$ where $\bmu(\cdot)$ is a possibly non-linear function and $\bvarepsilon_k$ is i.i.d. but possibly non-Gaussian noise, will have a distribution that is in the generalized location family. 

\begin{theorem}\label{thm:1}
Let $\bD_1, \ldots, \bD_n$ be an i.i.d. sample and assume that $\bD_1$ has distribution 
 $f\of{\bd_1;\btheta}$ belonging to a generalized location family.  Then the inequality (\ref{eqn:oneSidedInequality}) holds.
\end{theorem}

\begin{proof}
The proof of the Theorem proceeds in two parts. The first part establishes that, for any parameters $\btheta$ and $\btheta_0$, the ambiguity function (\ref{eq:mdef}) satisfies: $m(\btheta_0,\btheta_0)\geq m(\btheta_0,\btheta)$ (Claim 1).  The second part establishes that, for $f(\bd;\btheta)$ in a generalized location family, $m(\btheta_0,\btheta_0)=m(\btheta,\btheta)$ (Claim 2).  Putting these two parts together implies
\begin{align*}
m(\btheta,\btheta)\geq m(\btheta_0,\btheta).
\end{align*}
The theorem then follows  upon specialization of this inequality to $\btheta=\hat{\btheta}$.

We recall the integral form for the mean function 
\begin{align}
m(\btheta_0,{\btheta})=&E_{\btheta_0}[\log f(\bD;\btheta)] \nonumber \\
=&\int f(\bd;\btheta_0)\log f(\bd;\btheta) \;d\mu(\bd) ,
\label{eq:gdefn}
\end{align}
and the identity $E_{\btheta_0}[\log f(\bD;\btheta)] =n E_{\btheta_0}[\log f(\bD_1;\btheta)]$, which follows from the i.i.d. assumption.

Claim 1 in this proof follows from the non-negativity property of the Kullback-Liebler (KL) divergence, a well known result in statistics and information theory \cite{kullback1997information}. For completeness, we give a self contained proof. Start with the expression:    
\begin{align}
m(\btheta_0,\btheta_0)-m(\btheta_0,\btheta)=-\int f(\bd;\btheta_0) \log\frac{f(\bd;\btheta)}{f(\bd;\btheta_0)}\; d\mu(\bd) \nonumber \\
\label{eq:KL}
\end{align} 
Now, using the elementary inequality $\log(u)\leq u-1$ and the fact that $\int f(\bd,\tilde\btheta) d\mu(\bd)=1$ for all $\tilde\btheta$, the right hand side of (\ref{eq:KL}) is non-negative. Therefore, using the definition (\ref{eq:gdefn}), 
this establishes the claim
\begin{align}
E_{\btheta_0}[\log f(\bD;\btheta_0)] \geq E_{\btheta_0}[\log f(\bD;\btheta)].
\label{eq:claim1}
\end{align}

Claim 2 of this proof is a direct result of $f(\bd_k;\btheta)$ being in the generalized location family. Specifically, 
\begin{align*}
m(\btheta_0,\btheta_0) =& \int f\of{\bd;\btheta_0} \log f\of{\bd;\btheta_0} \;d\mu(\bd)
\\
=& n\int f\of{\bd_1-g(\btheta_0)}\log f\of{\bd_1-g(\btheta_0)} \;d\mu(\bd_1)
\\
=& n\int f\of{\tilde\bd_1-g(\btheta)}\log f\of{\tilde\bd_1-g(\btheta)}\;d\mu(\tilde\bd_1)
\\
=& m(\btheta,\btheta)
\end{align*}
where the second equality comes from the generalized location family definition and the third equality follows from making the change of variable of integration $\tilde{\bd}_1=\bd_1+(g(\btheta_0)-g(\btheta))$. This establishes the Theorem. \qed
\end{proof}

\subsection{Reparameterized Embeddings}\label{sec:embed}

The detection performance of the one-sided and two-sided tests for the global maximum depends on how much shift a local maximum $\hat{\btheta}\in\bTheta$ causes in the distribution of the associated validation function (\ref{eq:validation}) introduced in Sec. \ref{sec:oneSidedTest}.  In this section, we explain how embedding the parameters into a higher dimensional parameter space and a modified validation function can increase this shift, leading to improved detection performance. 

The proposed approach can be viewed as analogous to the advantageous use of higher dimensional embeddings in other mathematical problems. For example, in mathematical imaging, the level-set method for image segmentation \cite{osher2004level,sethian1996fast} embeds a two-dimensional curve in the plane as a level set of a higher-dimensional parameterized surface. As another example, parameter expansion is applied to stabilize numerical solutions to non-linear differential equations \cite{wang2008nonlinear}. A similar approach is used in computational statistics for accelerating the convergence of parameter estimates in the iterative parameter expansion expectation-maximization (PX-EM) algorithm \cite{liu1998parameter}. Furthermore, in machine learning, the support vector machine (SVM) \cite{vapnik1998support} improves classification performance by representing the decision region in the native lower dimensions as a separating hyperplane in a much higher dimensional space.  In analogy to the above examples, the reparameterization embedding we propose will lead to significant improvements in testing for the global optimum.

We define the reparameterized embedding as follows. As above let the log-likelihood function $\ell\of{\btheta}$ be parameterized by a {\em native} parameter $\btheta \in \bTheta$.  Let $\btheta'\in \bTheta'$ be a fictitious parameter and define the expanded parameterization $\tilde{\btheta}=(\btheta, \btheta')$ living in $\tilde{\bTheta}=\bTheta\times \bTheta'$. The native parameter $\btheta$ is thus embedded in the cylinder set $\{\btheta\} \times \bTheta'$ of the higher dimensional space $\tilde{\bTheta}$.  Associated with this embedding, define the augmented log-likelihood function $\tilde\ell(\bd;\tilde\btheta)$ parameterized by $\tilde\btheta\in\tilde\bTheta$. 
We can link the embedded parameterization $\tilde\btheta=(\btheta,\btheta') \in \tilde\bTheta=\bTheta\times \bTheta'$ to the native parameterization $\btheta\in \bTheta$ by fixing $\btheta'$, which we can assume is equal to $\bzero$ without loss of generality. Thus the native parameterization is equal to a cross-section of the embedded parameterization $\{\tilde\btheta=(\btheta,\mbf{0}):\btheta\in \bTheta\}$ which gives  the relation between the native and augmented log-likelihood functions: $\ell(\bd;\btheta)=\tilde\ell(\bd;\btheta,\mbf{0})$.

\def\amax{{\mathrm{amax}}}
Define the difference
between the  augmented log-likelihood function and the native log-likelihood function:
\begin{align}
\tilde{g}(\bd;\btheta,\tilde\btheta) =\tilde{\ell}(\bd;\tilde\btheta) - \ell(\bd;\btheta), \;\; \btheta \in \bTheta, \tilde\btheta\in \bTheta \times \bTheta'.\label{eq:augdifference}
\end{align}

We first introduce a prototype for a new class of global maximum validation functions introduced in this section. This prototype, called the max gap function, is defined as the maximum of $\tilde g$ with respect to $\tilde\btheta$ over the restricted embedding space $\{\btheta\}\times \bTheta'$ 
\begin{align}
  G(\bd,\btheta)=\max_{\tilde\btheta\in\{\btheta\}\times \bTheta'}\tilde{g}\of{\bd;{\btheta},\tilde\btheta} .
  \label{eq:maxgap}
\end{align}
When $\btheta=\hat\btheta$ is a non-global maximum of the native log-likelihood function $\ell(\bd;\btheta)$,  $G(\bd,\hat\btheta)$  measures the gap between the augmented log-likelihood $\tilde\ell(\bd;(\hat\btheta,\btheta'))$ maximized over $\btheta'\in \bTheta'$ 
%
and $\ell(\bd;\btheta)$ evaluated at $\btheta=\hat{\btheta}$. We would like to design the reparameterized embedding to maximize this gap. A mean gap function would provide a criterion for doing this.    

For native parameters $\btheta_0,\btheta \in \bTheta$ define the max mean gap function of $\btheta\in \bTheta$ 
\begin{align}
m_G(\btheta_0,\btheta)=\max_{\btheta'\in \bTheta'}E_{\btheta_0}[\tilde{g}(\bD;\btheta, (\btheta,\btheta'))] .
\label{eq:maxmeangap}
\end{align}
The mean of the max gap function is lower bounded by the max mean gap function.  
\begin{theorem}
For any $\btheta_0$ and $\btheta$ in $\bTheta$ the max gap function (\ref{eq:maxgap}) and the max mean gap function (\ref{eq:maxmeangap}) are non-negative. Furthermore, $E_{\btheta_0}[G(\bD;\btheta)]\geq m_G(\btheta_0,\btheta)$ and $m_G(\btheta_0,\btheta_0)=0$.
\end{theorem}

\noindent{\em Proof}:
Non-negativity of $G(\bd;\btheta)$ follows from the relation $\max_{\btheta'\in \bTheta'} \tilde{g}(\bd;\btheta,(\btheta,\btheta'))\geq \tilde{g}(\bd;\btheta,(\btheta,\bm{0}))$, which is equal to zero  since  $\tilde\ell(\bd;(\btheta,\bm{0}))=\ell(\bd;\btheta)$.
An identical argument establishes non-negativity of $m_G(\btheta_0,\btheta)$. 
%
%
The stated inequality 
follows from  the fact that,  for any random variable $Y(u)$ depending on a non-random parameter $u$, 
$\max_u E[Y(u)]$ is no greater than 
$E[\max_u Y(u)]$.  

The proof that $m_G(\btheta_0,\btheta_0)=0$ uses an argument similar to what we used to establish Claim 1 in the proof of Theorem 1. 
As $E_{\btheta_0}[\ell(\bD;\btheta_0)]=m(\btheta_0,\btheta_0)$, 
$m_G(\btheta_0,\btheta_0)
=
\max_{\btheta'\in\bTheta'} E_{\btheta_0}[\tilde\ell(\bD;(\btheta_0, \btheta')]-m(\btheta_0,\btheta_0).$
$E_{\btheta_0}[\tilde\ell(\bD;(\btheta_0, \btheta'))]$ is of the form 
\begin{align*}
\int f_0(\bd;\btheta_0) \log f_1(\bd;\btheta_0, \btheta')d\mu(\bd),
\end{align*} 
where the densities $f_0$ and $f_1$ satisfy $f_1(\bd;\btheta_0,\mbf{0})=f_0(\bd,\btheta_0)$. Therefore,  invoking the non-negativity of the KL divergence, which led to inequality (\ref{eq:claim1}) in the proof of Theorem 1, the integral above takes its maximum over $\btheta' \in \bTheta'$ when $\btheta'=\bm{0}$ and by definition (\ref{eq:mdef}) this value is equal to $m(\btheta_0,\btheta_0)$. 
\qed  

If the prototype validation function $G(\bd;\hat\btheta)$ defined in (\ref{eq:maxgap})  were to be used to test a local maximum $\hat\btheta$,  Theorem 2  suggests that the function $m_G(\btheta_0,\btheta)$ defined in (\ref{eq:maxmeangap}) can be used to design reparameterized embeddings that induce the largest possible positive shifts in $G(\btheta_0,\btheta)$ when $\btheta$ deviates from the true parameter $\btheta_0$, which would be expected to occur if $\btheta=\hat\btheta$ were a non-global maximum of $\ell(\bd;\btheta)$. We describe a spectral embedding procedure below that implements such a design strategy.

While the prototype validation function $G$ in (\ref{eq:maxgap})  has the benefits of simplicity and the analytical lower bound in Thm 2, it has two deficiencies that motivate an alternative gap function, called $g$ and defined in (\ref{eq:augvalidation}) below. First, the maximization in $G$ is restricted to the sub-space of  embedded parameter values $\tilde\btheta =(\btheta,\btheta')\in\tilde\bTheta$ over which $\btheta=\hat\btheta$ is fixed. This restriction deprives the prototype of extra degrees of freedom, potentially reducing its sensitivity to non-global maxima. Second, the evaluation of $G$ requires a global maximization over $\btheta'\in\bTheta'$, which may be challenging in practice.

We next introduce an alternative gap function that overcomes these deficiencies: it expands the maximization in (\ref{eq:maxgap}) to the full embedding space $\tilde\bTheta=\bTheta \times \bTheta'$ but allows the global maximum to be replaced by a local maximum found by an iterative algorithm initialized at $\tilde\btheta=(\hat\btheta,\mbf{0})$.

\begin{definition}
Given a point $\btheta \in \bTheta$ and an iterative algorithm for finding local maxima of a function $f(\tilde\btheta)$, $\tilde\btheta\in \tilde\bTheta=\bTheta\times \bTheta'$,  the basin of attraction $S_f\of{\btheta} \subseteq \tilde\bTheta$  is defined as an open set containing the point $\tilde{\btheta}=(\btheta,\mbf{0})$ such that the algorithm converges to a local maximum of $f(\tilde\btheta)$ when initialized at that point. 
\end{definition}  

Let $\hat\btheta\in \bTheta$ be a local maximum of the native log-likelihood function $\ell(\bd;\btheta)$ and let $\hat{\tilde\btheta} \in S_{\tilde\ell}\of{\hat\btheta} \subseteq \tilde\bTheta$ be a local maximum of the augmented log-likelihood function $\tilde\ell(\bd;\tilde\btheta)$ in a basin of attraction containing $\tilde\btheta=(\hat\btheta,\bzero)$.   
As an alternative gap function to (\ref{eq:maxgap}) we define the {\em augmented validation function} 
\begin{align}
  g(\bd,\hat{\btheta})=\max_{\tilde\btheta \in S_{\tilde\ell\,}\of{\hat\btheta}}\tilde{g}\of{\bd;\hat{\btheta},\tilde\btheta}.
  \label{eq:augvalidation}
\end{align}
This measures the gap between the value of the native log-likelihood function evaluated at $\hat{\btheta}$ and the maximum of the augmented log-likelihood function in the neighborhood of $(\hat{\btheta}, \mbf{0})$. 

As an analog to (\ref{eq:maxmeangap}), for native parameters $\btheta_0,\btheta \in \bTheta$, define the locally maximized mean of this gap function, called  the {\em augmented ambiguity function}
\begin{align}
m_g(\btheta_0,\btheta)=\max_{\tilde\btheta \in S_\mu\of{\btheta}}E_{\btheta_0}[\tilde{g}(\bD;\btheta, \tilde\btheta)] ,
\label{eq:augambiguity}
\end{align}
where for a point $\btheta\in \bTheta$,  $S_\mu\of{\btheta}$ denotes a basin of attraction of the mean function $\mu(\tilde\btheta)=E_{\btheta_0}[\tilde\ell(\bD;\tilde\btheta)]$ containing the point $\tilde\btheta=(\btheta,\bzero)$. 
\begin{theorem}
For any $\btheta_0$ and $\btheta$ in $\bTheta$ 
\begin{align}
0\leq m_g(\btheta_0,\btheta) \leq E_{\btheta_0}\left[ \max_{\tilde\btheta \in S_{\mu}(\btheta)}\tilde{g}\of{\bD;\btheta,\tilde\btheta}\right] 
.
\label{eq:thm2reln}
\end{align}
The leftmost inequality is achieved with equality when $\btheta=\btheta_0$. 
\end{theorem}

\noindent{\em Proof}:
The proof is similar to the proof of Theorem 2 and relies on     
the fact that $(\btheta, \mbf{0})\in S_\mu(\btheta)$, by definition of the bassin of attraction $S_\mu(\btheta)$. In particular, this implies that
%
$m_g(\btheta_0,\btheta) \geq E_{\btheta_0}[\tilde\ell(\bD;(\btheta,\mbf{0}))-\ell(\bD;\btheta)]=0,
$
%
establishing the leftmost inequality in (\ref{eq:thm2reln}). The condition for equality in the leftmost inequality follows from the fact that  $E_{\btheta_0}[\tilde\ell(\bD;\tilde\btheta)]\leq m(\btheta_0,\btheta_0)$, achieving equality when $\tilde\btheta=(\btheta_0,\bzero)$.
The rightmost inequality in (\ref{eq:thm2reln}) follows from the same property of expectation of maximized random variables as used in the proof of Theorem 2.  \qed   

Under the condition that the basins of attraction $S_{\tilde\ell}(\hat\btheta)$  and $S_\mu(\hat\btheta)$ are equal, Theorem 3 would give the lower bound $m_g(\btheta_0,\hat\btheta)$ on the mean shift in $g(\bd,\hat\btheta)$. If the augmented log-likelihood is smooth and the data  $\bD=\{\bD_1,\ldots, \bD_N\}$ consists of  i.i.d. samples, one can expect this condition to be satisfied asymptotically in $N$ since by the law of large numbers $\tilde\ell(\bD,\tilde\btheta)$ converges almost surely to its mean $E_{\theta_0}[\tilde\ell(\bD,\tilde\btheta)]$.  Under such conditions, similarly to Theorem 2 for the gap function $G$ (\ref{eq:maxgap}), Theorem 3 can be used to justify the use of the mean function $m_g$ to explore candidate reparameterized embeddings using the augmented validation function $g$ (\ref{eq:augvalidation}). 

In analogy to (\ref{eqn:rightTest}), for gap functions $G$ or $g$, we propose  a one sided reparameterized embedding test of $H_0: \hat{\btheta} = \hat{\btheta}_{Global}$ vs. $H_1: \hat{\btheta}\not=\hat{\btheta}_{Global}$. For the augmented validation function $g$ the proposed test is of the form: 
\begin{align}
\frac{g\of{\bd,\hat\btheta} -m_g\of{\hat\btheta,\hat\btheta} }
{\sqrt{v_g\of{\hat\btheta}}} \;\; \overset{H_1}{\underset{H_0}\gtlt}\;\; \tau,
\label{eqn:relaxedTest}
\end{align}
and similarly for the max gap function $G$.
As above, for $\btheta_0,\btheta_1\in \bTheta$, $m_g\of{\btheta_0,\btheta_1} = \E[\btheta_0]{g(\bD,\btheta_1)}$  and we have defined the variance $v_g(\btheta_0)=\Var[\btheta_0]{g\of{\bD,\btheta_0}}$.

In the following sections we show examples of how a well chosen reparameterized embedding space can result in significant improvement of global maximum testing. Given a statistical model and a candidate embedding $\tilde{\btheta}\in\tilde\bTheta$. Below we specify a computational procedure for selecting a reparameterized embedding space $\tilde\bTheta$. The procedure is inspired by Rao's locally optimal test \cite{cox:1974,rao:1948} of the hypotheses  $H_0:\btheta=\btheta_0$ vs. $H_1: \btheta=\btheta_0+\bdelta$, where $\bdelta$ is a small local perturbation of magnitude $\|\bdelta\|=\alpha$. Rao's score test solves a generalized eigenvalue problem of the form $\max_{\|\bu\in\bTheta\|} {\left|\bu^T\bs(\bd,\btheta_0)\right|^2}/{\bu^T \bI(\btheta_0) \bu}$, where $\bs(\bd,\btheta)=\nabla_{\btheta} \ell(\bd;\btheta)$  is the score function and $ \bI(\btheta)=-E_{\btheta}[\nabla_{\btheta}^2 \ell(\bd,\btheta)]$ is the Fisher information matrix. The solution of the generalized eigenvalue problem gives the optimal direction vector $\bu=\bI^{-1/2}(\btheta_0)\bs(\bd,\btheta_0)$ 
specifying the locally optimal perturbation as $\bdelta=\alpha\bu$.  

\vspace{0.1in}

\noindent{\bf Spectral embedding procedure}: 
Motivated by Theorems 2 and 3 and Rao's test, we propose a heuristic procedure to improve on a randomly initialized embedding space using a singular value decomposition (SVD) on a set of Rao locally optimal direction vectors that are computed offline. The effectiveness of this procedure is demonstrated in the numerical results sections below. 

First, an $\epsilon$-net of sampled parameters $\btheta_0^{(i)}$, $i=1,\ldots, p$, is constructed on $\bTheta$. For each $\btheta_0^{(i)}$ the non-global local maxima $\{\hat{\btheta}^{(i,j)}\}_{i,j=1}^{p,q_i}$ of $m(\btheta_0^{(i)},\btheta)$ are found. The likelihood function is reparameterized  into an {\em initial}  $\tilde\bTheta$  space of higher dimensional and the mean  $m(\btheta_0^{(i)},\tilde\btheta)=E_{\btheta_0^{(i)}}[\tilde\ell(\bD;\tilde\btheta)]$ of the augmented log-likelihood function is computed.   Let $\tilde{\btheta}^{(i,j)} \in \bTheta \times \bTheta^{'}$ be the point in the initial space associated with $\hat{\btheta}^{(i,j)}$ for each $i=1,\ldots,p,\ j=1,\ldots,q_i$.  Then the solutions $\{\bu(\tilde\btheta^{(i,j)})\}_{i,j}$ of the generalized eigenvalue problem are arranged as the columns of a matrix $\bG$, whose left principal singular-vector $\br$ is taken as a basis for the final spectral embedding subspace in $\bTheta'$. If more than one additional embedding dimension is desired, this procedure can be repeated with previously identified embedded subspaces successively removed.

\subsection {Example: Sinusoidal frequency estimation} \label{sec:sinusoidExample}

We return to the sinusoid in additive Gaussian noise example presented in Section \ref{sec:simpleModelD} to illustrate the theory presented in the previous section and to demonstrate the advantages of the proposed one-sided version of the Biernacki test (\ref{eqn:rightTest}) and the one-sided reparameterized embedding test (\ref{eqn:relaxedTest}).  
%

We embed the one dimensional frequency parameter $\theta_0\in \mathbb{R}$ into the expanded parameter $\tilde{\btheta}=(\theta_0,\theta_1,\ldots, \theta_k) \in \mathbb{R}^{k+1}$ and link $\tilde\btheta$ to the augmented log-likelihood through the augmented mean function $\tilde\bmu \in \mathbb R^N$
\begin{align}
\tilde\bmu\of{\tilde\btheta}=\sin{\theta_0 \bx + \theta_1 \bx^2 + ... + \theta_k \bx^{k+1} }.
\label{eq:sinumodel}
\end{align}
This reparameterization embedding introduces variation into the instantaneous frequency of the sinusoid, which enhances the gap between the native and augmented log-likelihood functions at the local maxima, leading to improved detection performance using the test (\ref{eqn:relaxedTest}).  With this embedding, the max gap function $G$ defined in (\ref{eq:maxgap})  takes the form
\begin{align} 
&G(\bd;\hat{\theta}) = \label{eq:reparembedding} \\
=&\frac{1}{2\sigma^2}\sum_{i=1}^n\left(\|\bmu(\hat{\theta})-\bd\|^2  -\min_{\theta_1, \ldots, \theta_k} \| \tilde\bmu(\hat{\theta}, \theta_1, \ldots, \theta_k)-\bd\|^2\right).
\nonumber
\end{align}
The max mean gap function $m_G(\theta_0, \hat{\theta})$ (\ref{eq:maxmeangap}), which,  by Theorem 2, lower bounds the mean shift in $G(\bd;\hat{\theta})$, has the representation:
\begin{align*}
m_G(\theta_0, \hat{\theta})=\frac{n}{\sigma^2}\left(F(\theta_0,\hat\theta) -
\cos((\theta_0-\hat\theta) x)\right),
\end{align*}
where, for large $N$,
\begin{align*}
&F(\theta_0,\hat\theta) = \\
	&\hspace{7mm}\max_{\theta_1, \ldots, \theta_k}\int_{0}^T \sin(\theta_0x)\sin(\hat\theta x+\theta_1 x^2 +\ldots +\theta_{k} x^{k+1}) dx 
\end{align*}
is the maximum of a Fresnel integral.  The forms $g$ and $m_g$ for the augmented validation function (\ref{eq:augvalidation}) and the ambiguity function (\ref{eq:augambiguity}) are similar except that the minimizations include $\theta_0$ and a local minimization is performed. 

To simplify the discussion we will respectively reparameterize the augmented and native log-likelihood functions by the associated mean functions: $\tilde\bmu(\tilde\btheta)$, defined in (\ref{eq:sinumodel}), and  $\bmu(\theta)=\tilde\bmu(\theta,\bzero)$, which lie in different subspaces of $\mathbb R^N$.  With this reparameterization the native log-likelihood $\ell(\bd;\theta)$ becomes $\ell(\bd;\bmu)$, with $\bmu=\sin(\theta \bx)$ lying in a one dimensional subspace while the augmented log-likelihood $\tilde\ell(\bd;\tilde\btheta)$ becomes $\tilde\ell(\bd;\tilde\bmu)$, with $\tilde\bmu$ lying in a higher dimensional subspace. In particular, the spectral embedding procedure, described at the end of the previous subsection, yields a mean vector $\tilde\bmu$ lying in the two dimensional subspace  $\{\tilde\bmu:\tilde\bmu=\sin(\theta_0 \bx)+\theta_1 \br, \theta_0, \theta_1 \in \mathbb R\}$, where $\br$ is the left principal singular vector of a $N\times \(\sum_{i=1}^p q_i\)$. The obtained vector $\br$ is shown in Figure~\ref{fig:relaxationDimension} for the case that the number of time samples is $N=100$, the measurement time interval is $T=1$, the Gaussian noise variance is $\sigma^2=1$ and the local maxima search region $\theta \in [0,4\pi]$.   This particular reparameterized embedding vector $\br$ defines the embedded space used in all simulations described in this subsection.  

\begin{figure}[ht]
\centering
\includegraphics[width=.9\columnwidth]{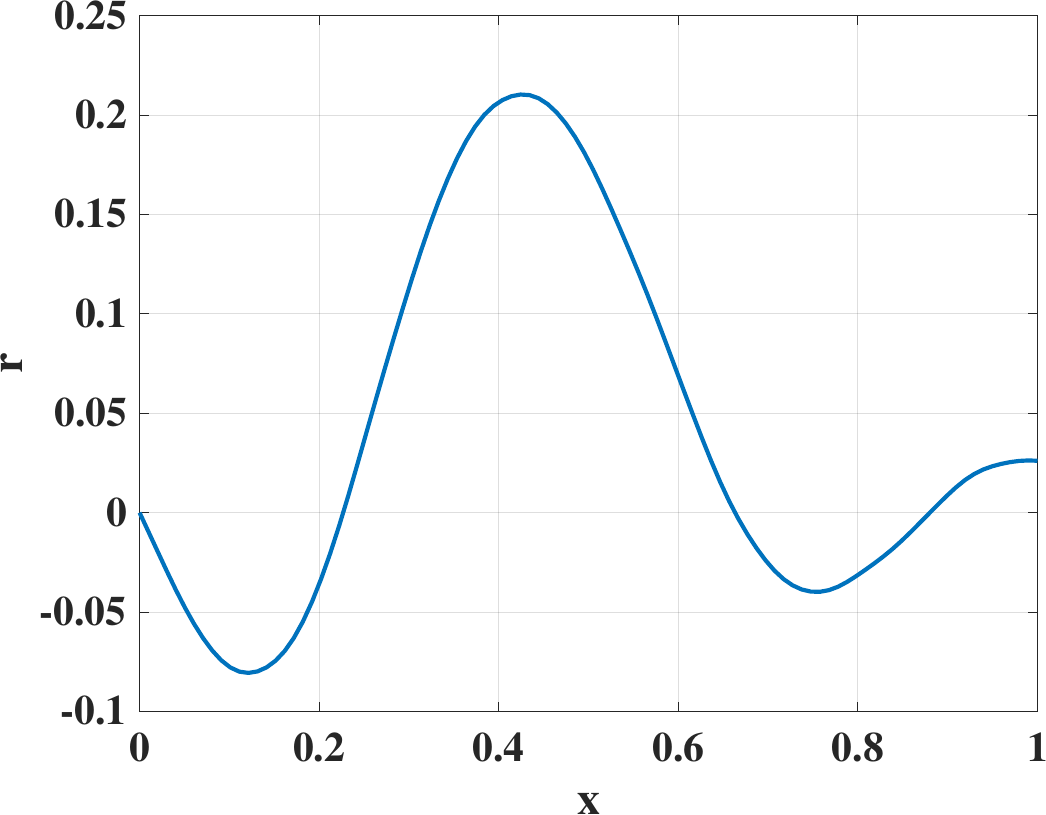}\hspace{1em} 
\caption{The principal left singular vector $\br$ computed by implementing the spectral procedure for learning embeddings, described at the end of Sec. \ref{sec:embed}, for the sinusoid in Gaussian noise example. The learned reparameterized embedding is the two dimensional subspace $\{\tilde\bmu \in \mathbb R^N:\tilde\bmu=\sin(\theta_0\bx)+\theta_1\br, \theta_0, \theta_1 \in \mathbb R\}$.
}
\label{fig:relaxationDimension}
\end{figure}


Figure~\ref{fig:relaxationMean} gives a graphical depiction of the proposed reparameterized embedding method in the context of this example. 
The left panel of the figure represents the case where the hypothesis $H_0:\hat\theta=\theta_{Global}$ is rejected, while in the right panel $H_0$ is not rejected. The figure shows the likelihood trajectories in the native parameterization $\bmu(\theta)$ as a green curve labeled $\mathcal U$.  The corresponding likelihood surface in the embedded parameterization $\tilde{\bmu}(\tilde\btheta)$ is shown as the disk labeled $\mathcal {\tilde U}$.    

%
%
The left and right panels in Figure \ref{fig:relaxationMean} each show two local maxima, a local maximum proximal to the true parameter $\theta_0$, denoted $\hat\theta_0$, and a local maximum distant from the true parameter, denoted $\hat\theta$. Associated with these local maxima are the expanded parameters $\tilde\btheta_0$ and $\hat{\tilde{\btheta}}$, respectively, that are depicted in basins of attraction $S\of{\hat\theta}$ of the augmented log-likelihood $\max_{\tilde\btheta}\tilde\ell\of{\bd, \bmu\of{\hat\theta, \tilde\btheta}}$.  On the right panel the proximal local maximum $\hat\theta_0$ is close to the true parameter $\theta_0$ so that the means $\mu\of{\hat\theta_0}$ and $\tilde\mu\of{\hat{\tilde\btheta}_0}$ are close to each other.
By contrast,  on the left panel these two means are not close to each other when the local maximum $\hat\theta$ is distant from $\theta_0$. The reparameterized embedding enhances this contrast by increasing the gap between the augmented log-likelihood and the native log-likelihood.  

\begin{figure}[ht]
\centering
\includegraphics[width=\columnwidth]{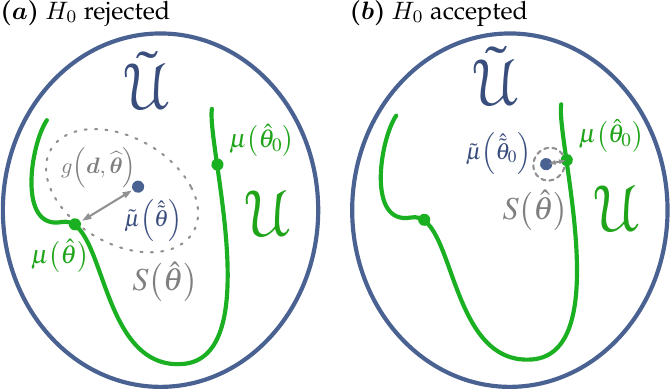}
\caption{(a) Reparameterized embedding depicted for the case where the augmented validation function $g$ (\ref{eq:augvalidation}) is used for testing for a global maximum. When initialized at $\tilde{\btheta}=\(\hat\btheta,\mbf{0}\)$ an iterative algorithm converges to a local maximum  $\hat{\tilde{\btheta}}$ in the basin of attraction $S\of{\hat\btheta}$ of $\hat\btheta$.  The large gap $g\of{\bd,\hat\btheta} = \ell\of{\bd,\tilde\bmu\of{\hat{\tilde\btheta}}} - \ell\of{\bd,\bmu\of{\hat\btheta}}$ causes a rejection of the null-hypothesis $H_0:\hat{\btheta}=\hat{\btheta}_0$.  (b) A maximum $\hat\btheta_0$ in the neighborhood of the true solution $\btheta_0$ leads to a smaller gap in $g$, and the null-hypothesis is accepted. }
\label{fig:relaxationMean}
\end{figure}

Figure~\ref{fig:relaxation2} shows the surface corresponding to the (negative) mean augmented log-likelihood,    $\{-E_{\theta_0}[{\ell}(\bd;\bmu(\tilde\btheta))]\}_{\tilde\btheta \in \tilde\Theta}$ 
as a surface over $\tilde\theta=(\theta,\theta^{'})$ along with its cross-section (blue curve) along the line $\{(\theta,0): \theta\in \Theta\}$,  which is the (negative) native mean log-likelihood trajectory $\{-E_{\theta_0}[\ell(\bd,\bmu(\theta))]\}_{\theta\in\Theta}$. The mean augmented log-likelihood has two local maxima, a global maximum $\hat{\tilde\btheta}_0=(\hat\theta_0,\hat\theta'_0)$ near $(\theta_0,0)$ and a local maximum $\hat{\tilde\btheta}=(\hat{\theta}, \hat{\theta}')$. The green curve traces out $-\max_{\theta'\in \Theta'} E_{\theta_0}[\ell(\bd,\tilde\bmu(\theta,\theta'))]$ and the gap between the blue and green curves is the augmented ambiguity function $m_g(\theta_0,\theta)$ (\ref{eq:augambiguity}). 

\begin{figure}[ht]
\centering
\includegraphics[width=0.9\columnwidth]{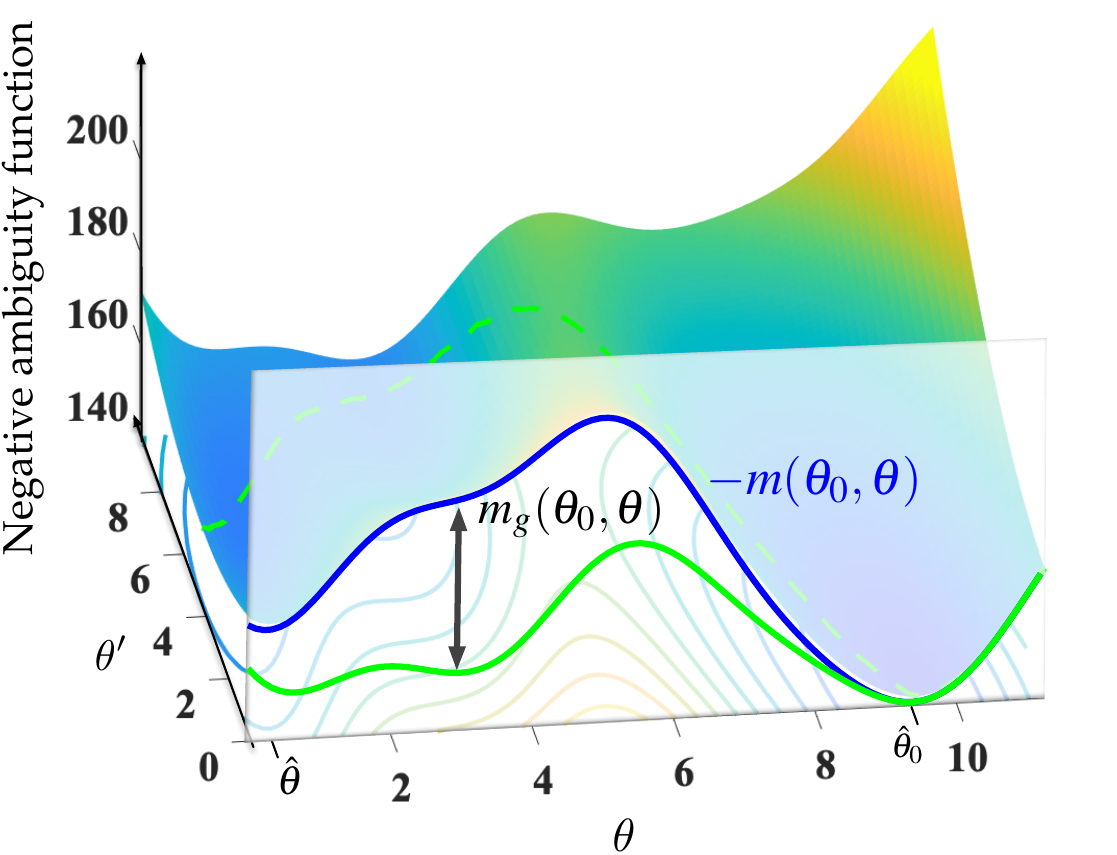}
\caption{The (negative) mean of the augmented log-likelihood surface associated with the sinusoid in additive Gaussian noise example given by (\ref{eqn:simpleModelA}) with reparameterized two dimensional embedding constructed from the principal singular vector $\br$ shown in Fig. \ref{fig:relaxationDimension}. The gap between the blue and green curves is the augmented ambiguity function $m_g(\theta_0,\theta)$ (\ref{eq:augambiguity}).
}
\label{fig:relaxation2}
\end{figure} 
\begin{figure}[ht]
\centering
\includegraphics[width=.9\columnwidth]{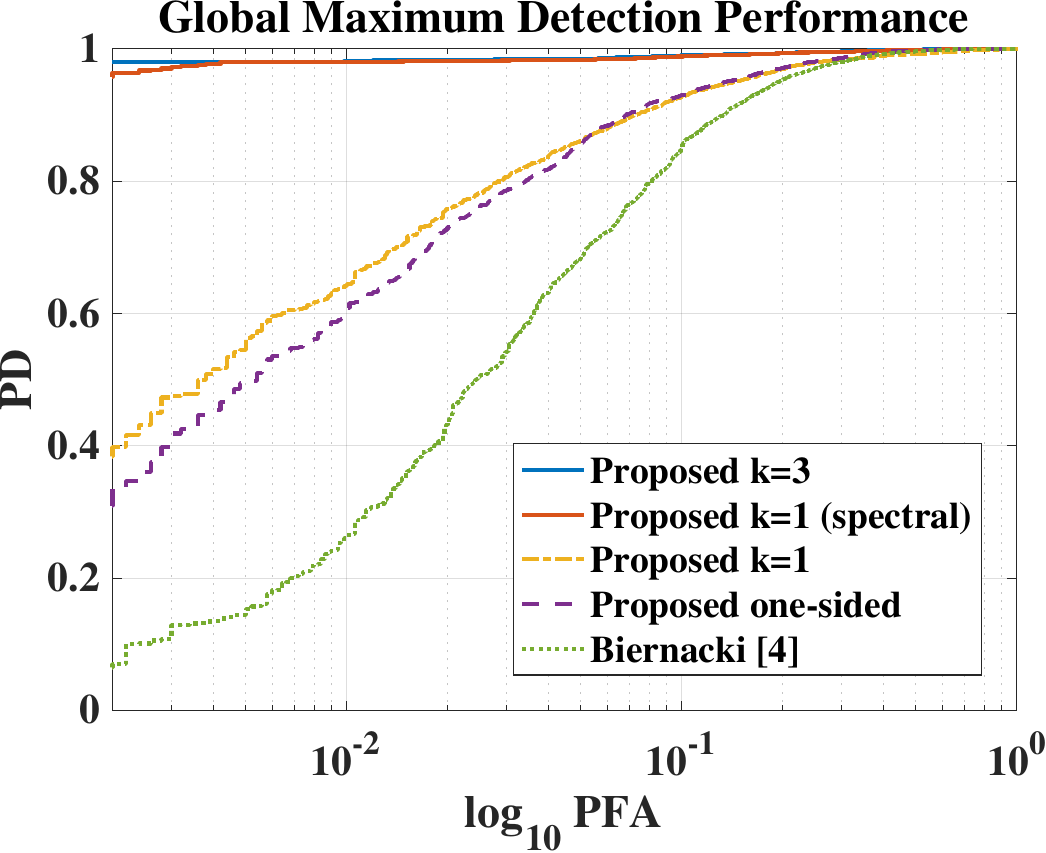}
\caption{Empirically estimated receiver operating characteristic (ROC) curve of probability of detection (PD) versus probability of false alarm (PFA) of the proposed one-sided test (\ref{eqn:rightTest}) and the proposed  reparameterized embedding test (\ref{eqn:relaxedTest}) with two and four embedding dimensions ($k=1$ and $k=3$)  as compared to Biernacki two sided test (Biernacki [4]) for frequency estimation of a sinusoid in additive Gaussian noise. The curve labeled ``Proposed $k=1$ (spectral)'' corresponds to the test (\ref{eqn:relaxedTest}) implemented with the two dimensional spectral embedding procedure using the principal direction vector $\br$ shown in Figure \ref{fig:relaxationDimension}.
}
\label{fig:relaxationNaiveExample}
\end{figure} 

Figure~\ref{fig:relaxationNaiveExample} demonstrates the improvements achieved using the proposed reparameterized embedding test (\ref{eqn:relaxedTest})  
with respect to testing accuracy, as quantified by receiver operating characteristic (ROC) curves. The ROC curve sweeps out the probability of detection (PD) against the probability of false alarm (PFA) achieved by a test that a local maximum is the global maximum. Tests with higher ROC curves are more accurate.  10,000 simulations were performed with the following parameters: a true sinusoidal frequency $\theta_0=3\pi$ and, $N=100$, $T=1$, $\sigma^2=1$, and search region $\theta\in [0,4\pi]$. In order of increasing accuracy these tests are: the two-sided Biernacki test \cite{biernacki:2005}, the one-sided version of the Biernacki test (\ref{eqn:rightTest}), equivalent to the proposed test with $k=0$, followed by the proposed one-sided with the embedding given by (\ref{eq:sinumodel}) for $k=1$ and $k=3$ respectively.  Using $k=3$ or the one-sided test based on a single embedding basis chosen according to the proposed spectral approach, achieves significantly better performance than Biernacki tests.  For all of the reparameterized embedding tests the maximization in the augmented validation function (\ref{eq:reparembedding}) was computed numerically using a limited-memory quasi-Newton solver \cite{nocedal:1999}.

\section{Application to Wavefront Sensing}\label{sec:wavefront}


We apply the proposed global maximum testing framework to the problem of jointly estimating camera blur and pose from a known calibration target in the presence of aliasing, studied in \cite{leblanc:2018}. Wavefront phase-aberrations characterize the camera's point spread function (PSF) through deviations of an otherwise ideal optical system, and a Zernike polynomial basis \cite{Zernike:1934} is used to parameterize these phase aberrations resulting in a parametric blur model. The solution to the inverse problem is the global maximum of a highly nonconvex log-likelihood function, that admits many local maxima \cite{barakat:1992,isernia:1995,gerwe:2008,moretta:2019}.  The inverse problem is typically solved using iterative optimization algorithms that converge to local maxima. Thus, in this context, a test for global maximum is a test of global convergence of the iterative algorithm. We demonstrate how the proposed test of global convergence can be used to reduce the computation burden. 


The spectral embedding procedure described in Section~\ref{sec:embed} is used in conjunction with the one-sided test (\ref{eqn:relaxedTest}) to substantially improve the detectability of convergence of the estimated blur to a suboptimal maximum of the log-likelihood function. Under the generalized imaging model \cite{goodman:1996}, the PSF $h$ is described in terms of a non-negative aperture function $A$ and a real-valued phase function $\Psi$, expressed as a linear combination of Zernike polynomials, where the parameter vector $\btheta$ contain the basis coordinates. Specifically, $h$ is given as
\begin{align*}
	h\of{x,y;\btheta} &= c_0\abs{g\of{w_x,w_y;\btheta}}^2  
	\\
	g\of{w_x,w_y;\btheta} &= \mathcal{F}^{-1}\of[]{A\of{w_x,w_y}\exp\[\Psi\of{w_x,w_y;\btheta}\]},
\end{align*}
where $c_0$ is a normalizing constant that ensures the PSF integrates to 1, $g\of{w_x,w_y;\btheta}$ is the coherent transfer function (CTF), and $\mathcal{F}^{-1}$ is the inverse Fourier transform. A test for global convergence is constructed using the proposed spectral method to embed this parametric blur model into the space of non-negative PSFs on a Nyquist-sampled grid. Monte Carlo simulations performed at a moderate signal-to-noise ratio (20 dB) and blur strength (0.025 waves RMS) demonstrate the substantial power of the proposed test. Over 100 such trials, a limited-memory quasi-Newton search started from the point representing no phase aberrations (an ideal imaging system) led to non-global local maxima 96\% of the time. Table~\ref{tbl:PSFTestPerformance} compares the power of the Biernacki's test \cite{biernacki:2005} with the proposed approach when both tests are operated at a false alarm rate of 0.01. The observed improvement is consistent with the simulation example described in Section~\ref{sec:sinusoidExample}.
\begin{table}[h]
\begin{center}
\begin{tabular}{l || c  c}
\bf Test							& \bf PFA    	& \bf PD 	\\
\hline
Biernacki \cite{biernacki:2005}		& $0.01$ 		& $0.22$	\\
Proposed (\ref{eqn:relaxedTest}) 	& $0.01$ 		& $1.0$
\end{tabular}
\end{center}
\caption{Empirical performance of global convergence tests for an iterative quasi-Newton maximum likelihood estimator of the camera point spread function (100 Monte-Carlo trials).  The probability of detection (PD) achieved by the proposed global maximum test (\ref{eqn:relaxedTest}) is perfect, while Biernacki's test [5] only attains $0.22$, when the probability of false alarm (PFA) is constrained to be $0.01$. }
\label{tbl:PSFTestPerformance}
\end{table}
%
 

Given a local maximum of the log-likelihood suspected of being a suboptimal solution, one would like to exploit knowledge of this local maximum to identify alternative regions of the parameter space likely to contain a better solution.  The PSF $\tilde{h}$ corresponding to a perturbed phase-screen $\Psi+\beta$ can be expressed as the modulus squared of a convolution of CTFs associated with $\Psi$ an $\beta$ respectively
\begin{align*}
\tilde h 	&= c_0\abs{ \mathcal{F}^{-1}\of[]{Ae^{j\Psi} A_Be^{j\beta}} }^2\\
	&= c_0\abs{ \mathcal{F}^{-1}\of[]{Ae^{j\Psi}} * \mathcal{F}^{-1}\of[]{A_Be^{j\beta}} }^2\\
	&= c_0\abs{ g * \tilde{g} }^2,
\end{align*}
where $A_B$ is the binary aperture corresponding to the support of $A$.  Letting $\[h\]_{m,n}$ be the $(m,n)^\textrm{th}$ element of a Nyquist sampled representation of the PSF, then a point-wise bound on the magnitude of the change induced by $\beta$ is given by
\begin{align}
\abs{[\epsilon]_{m,n}} =& \abs{ [h]_{m,n} - c_0\abs{[g]_{m,n}}^2 } \nonumber\\	
	\leq& \norm{\tilde{g}-a\delta}\[\norm{\tilde{g}-a\delta} + 2\frac{\abs{[g]_{m,n}}}{\norm{g}}\], \label{eqn:PSFBound}
\end{align}
where $\delta$ is the Kronecker delta function, and $a$ an arbitrary complex constant such that $\abs{a}=1$.  This point-wise bound on the PSF perturbation $\epsilon$, associated with the wavefront perturbation $\beta$, is minimized when $\angle a = \angle\[\tilde{g}\]_{0,0}$. Under this condition, the right-hand side of (\ref{eqn:PSFBound}) is monotonic in the Strehl ratio \cite{martial:1991} associated with $\beta$, which we will denote as $c_0\abs{\[\tilde{g}\]_{0,0}\of{\beta}}^2$. Thus, the set of wavefronts that maximize the Strehl ratio for a fixed root mean square (RMS) perturbation strength also minimizes the worst-case, point-wise error in the perturbed PSF.  These wavefronts are given by
\begin{align}
\{\beta = \argmax_{\tilde\beta}\ c_0\abs{\[\tilde{g}\]_{0,0}\of{\tilde\beta}}^2: \norm{\tilde\beta}^2=\tau\}.\label{eqn:maxStrehlSet}
\end{align}
A PSF $h$ perturbed by a wavefront in (\ref{eqn:maxStrehlSet}) will result in a new PSF $\tilde{h}$ that is point-wise close to $h$ despite its wavefront $\Psi+\beta$ being $\tau$ waves RMS from $\Psi$.  The proposed test for global convergence can be used in conjunction with this restarting strategy to search for globally optimal solutions.


A Monte Carlo simulation was performed to assess the efficacy of the proposed approach for identifying the global maximum of the likelihood function. A simulated imaging system \cite{leblanc:2018} was configured to provide moderate SNR images (20 dB), and the number of blur aberration parameters was varied to alter the difficulty of the resulting inverse problem. As the number of Zernike modes in the model increases, so does the probability of encountering local maxima.  
A limited-memory quasi-Newton search \cite{nocedal:1999} was used to identify stationary points of the log-likelihood starting from a diffraction-limited model. If the reparameterized embedding approach described in Section~\ref{sec:embed} failed to reject the null hypothesis at a false alarm level $\alpha=0.01$ the search was terminated, otherwise a new starting point 0.2 waves RMS away from the current maximum was chosen according to (\ref{eqn:maxStrehlSet}), and the search was continued. Figure~\ref{fig:optimizationMonteCarlo} shows the mean and standard errors of runtimes corresponding to 10 independent realizations of the same camera model. For comparison, the simulated annealing algorithm provided in \MATLAB \ Optimization Toolbox version 8.0 was used as a point of reference. The simulated annealing algorithm was provided the objective function gradients and was terminated according to an oracle criterion: terminate the first time that \emph{any local maximum fell within 0.01 waves RMS of the true solution}. In the astronomical imaging community, wavefront descriptions of optical systems typically include Zernike models up to at least radial-order 3 (7 Zernike modes). For models of such high complexity, the proposed reparameterized embedding strategy resulted in a five times reduction in total runtime. Figure~\ref{fig:PSFSearch} illustrates a typical sequence of PSFs associated with the global search procedure when 12 Zernike modes parameterize the blur. Despite the relatively small differences between the PSFs, the two non-global local maxima (local \#1 and local \#2) of the log-likelihood are associated with relatively large wavefront perturbations errors of 0.117 and 0.115 waves RMS, respectively.
\begin{figure}[ht]
\centering
\includegraphics[width=.9\columnwidth]{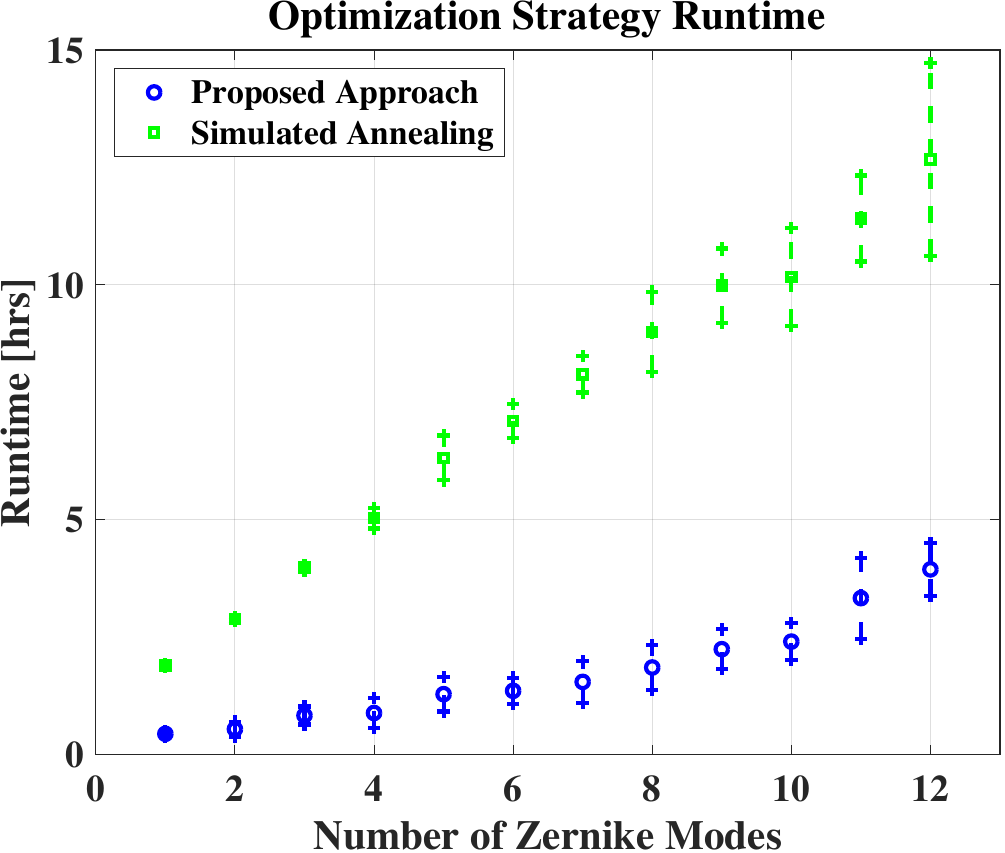}
\caption{Monte-Carlo study of optimization runtimes as a function of the number of aberration modes in the model. The proposed global convergence test reduces runtime by at least a factor of 5.}
\label{fig:optimizationMonteCarlo}
\end{figure} 
\begin{figure}[ht]
\centering
\includegraphics[width=\columnwidth]{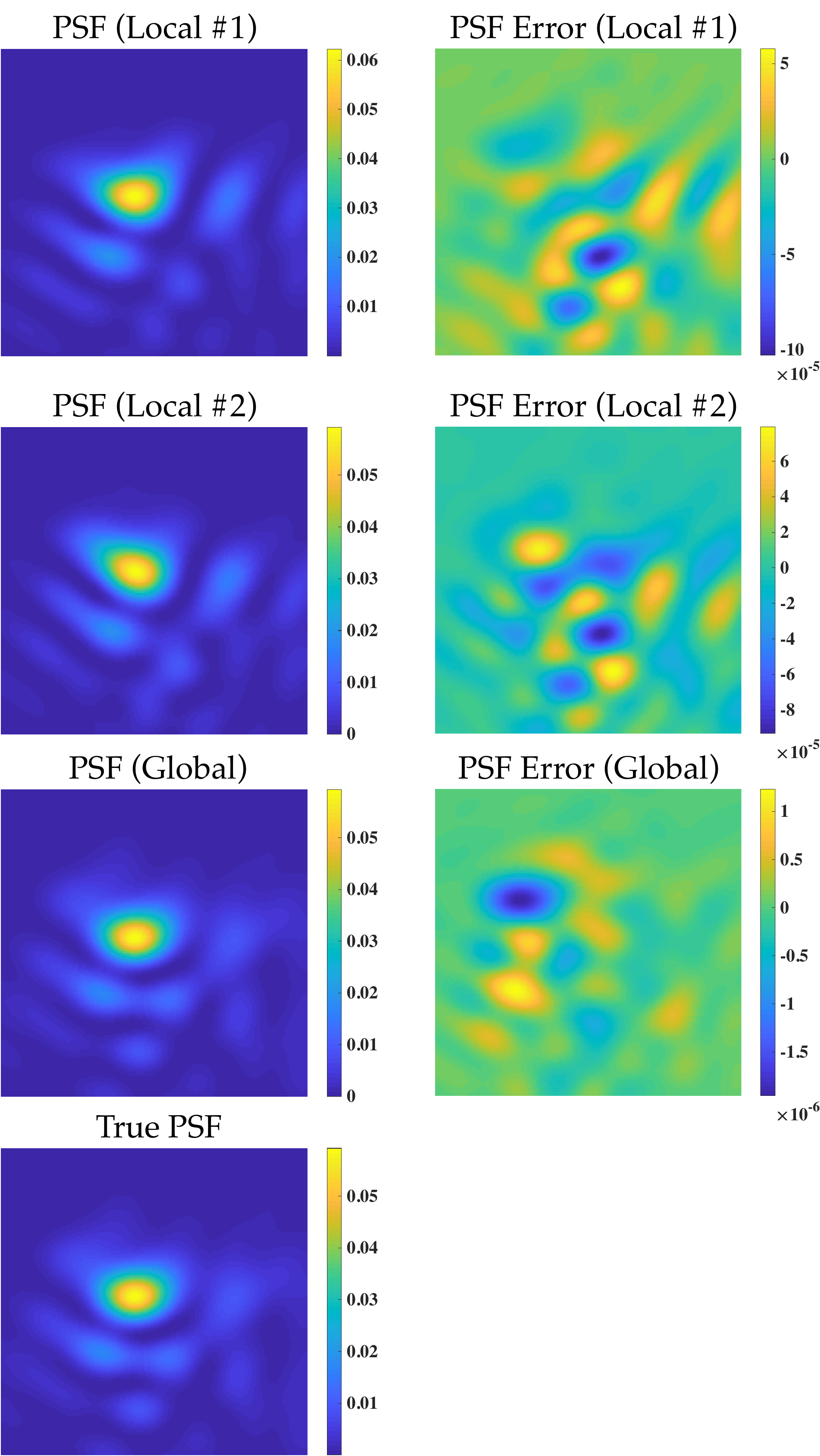}
\caption{A typical series of PSF estimates are shown alongside the true solution (left column) for local maxima of the log-likelihood function when 12 Zernike coefficients parameterize the PSF. The associated errors relative to the true solution are shown in the rightmost column.}
\label{fig:PSFSearch}
\end{figure}

\section{Concluding Remarks}\label{sec:conclude}

This paper addresses a principal computational bottleneck in non-convex imaging and vision problems: It determines if a local maximum found by a non-global optimization algorithm is a global maximum. Specifically, we introduced a powerful new method for validating that a local maximum is a global maximum when the objective function is specified as the likelihood function associated with a parametric statistical model. The proposed method implements a one-sided threshold test on a novel validation function defined as the normalized difference between the log-likelihood function and an augmented log-likelihood function, each evaluated at a local maximum point. The augmented log-likelihood is constructed by embedding the original parameter vector into a higher dimensional parameter space, a procedure we call reparameterized embedding, and the validation function is evaluated at the local maximum before thresholding. We proposed a computational spectral embedding procedure for identifying good reparameterized embeddings, and numerical results are presented exhibiting an extraordinarily high level of detection accuracy, e.g., achieving significantly better accuracy than the two-sided test proposed by Biernacki. Finally, to demonstrate how our results can dramatically impact non-convex imaging applications, we applied the proposed test to a set of local maxima generated from multiple restarts of an iterative maximum likelihood algorithm for reconstructing camera blur from images of a calibration target. When the test is used as a stopping rule, i.e., the restarts are stopped when the test declares a global maximum has been found, it reduced total runtime by a factor of five.

Code for reproducing the key figures from this document is available at \url{https://github.com/jwleblan/localMinima}.

\begin{acknowledgements}
This work was partially supported by ARO grant W911NF-15-1-0479 and a DOE NNSA grant to the University of Michigan Consortium on Verification Technology.  
\end{acknowledgements}

%
%

\bibliographystyle{spmpsci}      
\bibliography{extracted.bib}  	  

%
%

\end{document}